\documentclass[12pt]{article}
\usepackage{amsmath,amsfonts,amssymb,color,theorem,datetime2}

\newcommand{\bA}{\boldsymbol{A}}
\newcommand{\bB}{\boldsymbol{B}}
\newcommand{\bC}{\boldsymbol{C}}
\newcommand{\bF}{\boldsymbol{F}}
\newcommand{\bh}{\boldsymbol{h}}
\newcommand{\bH}{\boldsymbol{H}}
\newcommand{\bi}{\boldsymbol{i}}
\newcommand{\bj}{\boldsymbol{j}}
\newcommand{\bk}{\boldsymbol{k}}
\newcommand{\bone}{\boldsymbol{1}}
\newcommand{\bP}{\boldsymbol{P}}
\newcommand{\bq}{\boldsymbol{q}}
\newcommand{\bQ}{\boldsymbol{Q}}
\newcommand{\bS}{\boldsymbol{S}}
\newcommand{\bT}{\boldsymbol{T}}
\newcommand{\bV}{\boldsymbol{V}}
\newcommand{\bW}{\boldsymbol{W}}
\newcommand{\bX}{\boldsymbol{X}}
\newcommand{\bzD}{\boldsymbol{\Delta}}
\newcommand{\bzS}{\boldsymbol{\Sigma}}
\newcommand{\C}{\mathbb{C}}
\newcommand{\CQ}{\mathcal{Q}}
\newcommand{\CS}{\mathcal{S}}
\newcommand{\dn}{{\text{\tiny dn}}}
\newcommand{\hR}{\hat{R}}
\newcommand{\Hm}{\mathbb{H}}
\newcommand{\R}{\mathbb{R}}

\newcommand{\tcb}{\textcolor{blue}}

\newcommand{\tq}{\tilde{q}}
\newcommand{\up}{{\text{\tiny up}}}
\newcommand{\za}{\alpha}

\newcommand{\zf}{\varphi}
\newcommand{\zG}{\Gamma}
\newcommand{\zl}{\lambda}
\newcommand{\zL}{\Lambda}
\newcommand{\zo}{\omega}
\newcommand{\zO}{\Omega}
\newcommand{\zs}{\sigma}
\newcommand{\zS}{\Sigma}

\newtheorem{theorem}{\bf Theorem}
\numberwithin{theorem}{section}

\newtheorem{corollary}[theorem]{\bf Corollary}
\newtheorem{proposition}[theorem]{\bf Proposition}
{\theorembodyfont{\rmfamily}\newtheorem{example}[theorem]{\textbf Example}}
\newenvironment{proof}
{\begin{trivlist}\item[\hskip\labelsep\quad{\bf Proof.}
\hspace{0.5 em}]}{\hfill \rule{0.5em}{0.5em} \end{trivlist}}
\numberwithin{equation}{section}

\author{Francesco Demontis\footnote{Dipartimento di Matematica e Informatica,
Universit\`a di Ca\-glia\-ri, Via Ospedale 72, 09124 Cagliari, Italy. Email:
fdemontis@unica.it, cornelis110553@gmail.com}\,, Cornelis van der
Mee\footnotemark[2]}
\title{Quaternion Algebra Approach to Nonlinear\\ Schr\"odinger Equations with
Nonvanishing\\ Boundary Conditions\tcb{\footnote{LaTeX compilation date and
time: \tcb{\DTMnow}}}}

\begin{document}
\date{}
\maketitle

\begin{abstract}
In this article we apply quaternionic linear algebra and quaternionic linear
system theory to develop the inverse scattering transform theory for the
nonlinear Schr\"o\-din\-ger equation with nonvanishing boundary conditions. We
also determine its soliton solutions by using triplets of quaternionic matrices.
\end{abstract}

\section{Introduction}\label{sec:1}

The initial-value problem for the focusing nonlinear Schr\"o\-din\-ger (NLS)
equation
\begin{equation}\label{1.1}
i\tq_t+\tq_{xx}-2|\tq|^2\tq=0
\end{equation}
with nonvanishing boundary conditions $\tq(x,t)\to\tq_{r,l}(t)$ as
$x\to\pm\infty$, where $\tq_{r,l}(t)=\mu\,e^{-2i\mu^2t+i\theta_{r,l}}$ for
a positive constant $\mu$ and phases $\theta_{r,l}\in\R$, has been abundantly
studied using the inverse scattering transform (IST) technique
\cite{FT,BK,Federica2,BLM}. In \cite{BM} the IST with full account of the
spectral singularities has led to rogue wave solutions of the focusing NLS with
nonvanishing boundary conditions. Throughout this article we study instead of
\eqref{1.1} the NLS-like equation
\begin{equation}\label{1.2}
iq_t+q_{xx}-2|q|^2q+2\mu^2q=0,
\end{equation}
obtained from \eqref{1.1} by applying the gauge transformation
$$\tq(x,t)=e^{-2i\mu^2t}q(x,t),$$
where $q(x,t)$ tends to the time invariant limits
$q_{r,l}=\mu\,e^{i\theta_{r,l}}$ as $t\to\pm\infty$. We also write
$\CQ=\left(\begin{smallmatrix}0&q\\-q^*&0\end{smallmatrix}\right)$ to convert
\eqref{1.2} into the $2\times2$ matrix NLS-like equation
\begin{equation}\label{1.3}
i\zs_3\CQ_t+\CQ_{xx}-2\CQ^3-2\mu^2\CQ=0_{2\times2},
\end{equation}
where $\CQ^\dagger=-\CQ$. Here we write $I_p$ for the identity matrix of order
$p$, $0_{p\times r}$ for the $p\times r$ matrix with zero entries, the dagger
for the complex conjugate matrix transpose, and
$\zs_3=\left(\begin{smallmatrix}1&0\\0&-1\end{smallmatrix}\right)$ for the third
Pauli matrix. The nonlinear Schr\"o\-din\-ger equations have served as
mathematical models for surface waves on deep waters \cite{Abl,ASg,ZS}, signals
along optical fibers \cite{HT,Has,Shaw}, plasma oscillations \cite{Z}, magnetic
spin waves \cite{CTNP,ZP}, and particle states in Bose-Einstein condensates
\cite{PS1,PS2,KFCG}.

In \cite{DM22} a new method to solve the initial-value problem of the matrix
NLS equation by means of the inverse scattering transform technique was
introduced. Instead of determining the time evolution of the scattering data
associated with the Zakharov-Shabat system $v_x=(-ik\zs_3+\CQ)v$ and solving the
Marchenko integral equations associated with the time dependent scattering data
(as in \cite{Federica2}), we determined the time evolution of the
scattering data associated with the matrix Schr\"o\-din\-ger equation
$-\psi_{xx}+\bQ\psi=\zl^2\psi$, where $\bQ=\CQ^2+\CQ_x+\mu^2I_2$ and
$\zl=\sqrt{k^2+\mu^2}$ is the conformal mapping defined for all complex $k$ cut
along $[-i\mu,i\mu]$ and satisfying $\zl\sim k$ at infinity. Since this
conformal mapping $k\mapsto\zl$ is $1,1$ for $k$ in the upper half-plane $\C^+$
cut along $(i0,i\mu]$ and $\zl\in\C^+$, this has led to a great simplification
compared to the treatment based on the Zakharov-Shabat system
$v_x=(-ik\zs_3+\CQ)v$ given in \cite{BK,Federica2,BLM}.

In this article we restrict ourselves to solving the initial-value problem for
the $1+1$ focusing NLS equation. The advantage of this restriction is that the
potential $\bQ$ satisfies the symmetry relation
\begin{equation}\label{1.4}
\bQ^*=\zs_2\bQ\zs_2,
\end{equation}
where the asterisk denotes complex conjugation without transposition and
$\zs_2=\left(\begin{smallmatrix}0&-i\\i&0\end{smallmatrix}\right)$ is the second
Pauli matrix. Using the algebra isomorphism between the algebra $\bzS$ of
complex $2\times2$ matrices $S$ satisfying $S^*=\zs_2S\zs_2$ and the division
ring $\Hm$ of quaternions \cite{Ham3}, we can reduce the resolution of the
Marchenko integral equations to solve the inverse scattering problem for the
matrix Schr\"o\-din\-ger equation $-\psi_{xx}+\bQ\psi=\zl^2\psi$ to calculations
involving quaternions.

In this article we rely significantly on the direct and inverse scattering
theory for the matrix Schr\"o\-din\-ger equation developed
for $\bQ^\dagger=\bQ$, in \cite{AM,AW18,AW20} on the half-line and in
\cite{WK,MO,AKV} on the full line, albeit with some modifications due to the
symmetry relation \eqref{1.4}. For technical reasons we assume throughout this
article that the integral $\int_{-\infty}^\infty dx\,(1+|x|)\|\bQ(x)\|$
converges. For the various applications of the matrix Schr\"o\-din\-ger equation
with selfadjoint potential we refer to \cite{AW20}.

Let us discuss the contents of this article. In Sec.~\ref{sec:2} we review the
direct and inverse scattering theory of the matrix Schr\"o\-din\-ger equation
with symmetry relation \eqref{1.4}, where we essentially rely on the more
general scattering theory given in \cite{DM21,DM22}. In Sec. \ref{sec:3} we
discuss the time evolution of the scattering theory. In Sec. \ref{sec:4} we
discuss matrices having quaternion elements and their isomorphic images of
double matrix order. Here we rely on the seminal monograph on quaternionic
matrices by Rodman \cite{R3}. Section \ref{sec:5} is devoted to the multisoliton
solutions of the AKNS system with nonvanishing boundary conditions parametrized
by choosing minimal triplets of quaternionic matrices. Results on the
invertibility of the Sylvester solutions $\bP_r$ and $\bP_l$ appearing in the
multisoliton solutions are relegated to Appendix~\ref{sec:A}.

\section{Direct and Inverse Scattering}\label{sec:2}

In this article we discuss the direct and inverse scattering theory for the
matrix Schr\"o\-din\-ger equation
\begin{equation}\label{2.1}
-\psi_{xx}+\bQ\psi=\zl^2\psi,
\end{equation}
where the complex $2\times2$ potential $\bQ$ satisfies the symmetry relation
\begin{equation}\label{2.2}
\bQ^*=\zs_2\bQ\zs_2
\end{equation}
and hence belongs to the algebra
$\bzS=\left\{\left(\begin{smallmatrix}S_1&-S_2^*\\S_2&S_1^*\end{smallmatrix}
\right):S_1,S_2\in\C\right\}$. Then this potential $\bQ$ also satisfies the more
restrictive adjoint symmetry relation
\begin{equation}\label{2.3}
\bQ^\dagger=\zs_3\bQ\zs_3,
\end{equation}
where $\zs_3=\left(\begin{smallmatrix}1&0\\0&-1\end{smallmatrix}\right)$ is the
third Pauli matrix. Hence, by virtue of \eqref{2.3}, all of the results on the
direct and inverse scattering theory of \eqref{2.1} developed in
\cite{DM21,DM22} go though in the present situation, although we need to discuss
the impact of the more restrictive symmetry relation \eqref{2.2} on the results
separately.

Let us define the {\it Jost solution from the left} $F_l(x,\zl)$ and the
{\it Jost solution from the right} $F_r(x,\zl)$ as those solutions of the matrix
Schr\"o\-din\-ger equation \eqref{2.1} which satisfy the asymptotic conditions
\begin{subequations}\label{2.4}
\begin{alignat}{3}
 F_l(x,\zl)&=e^{i\zl x}\left[I_2+o(1)\right],&\qquad&x\to+\infty,\label{2.4a}\\
 F_r(x,\zl)&=e^{-i\zl x}\left[I_2+o(1)\right],&\qquad&x\to-\infty.\label{2.4b}
\end{alignat}
\end{subequations}
Calling $m_l(x,\zl)=e^{-i\zl x}F_l(x,\zl)$ and $m_r(x,\zl)=e^{i\zl x}F_r(x,\zl)$
{\it Faddeev functions}, we easily define them as the unique solutions of the
Volterra integral equations
\begin{subequations}\label{2.5}
\begin{align}
m_l(x,\zl)&=I_2+\int_x^\infty dy\,\frac{e^{2i\zl(y-x)}-1}{2i\zl}\bQ(y)
m_l(y,\zl),\label{2.5a}\\
m_r(x,\zl)&=I_2+\int_{-\infty}^x dy\,\frac{e^{2i\zl(x-y)}-1}{2i\zl}\bQ(y)
m_r(y,\zl).\label{2.5b}
\end{align}
\end{subequations}
Then, for each $x\in\R$, $m_l(x,\zl)$ and $m_r(x,\zl)$ are continuous in
$\zl\in\C^+\cup\R$, are analytic in $\zl\in\C^+$, and tend to $I_2$ as
$\zl\to\infty$ from within $\C^+\cup\R$. For $0\neq\zl\in\R$ we can reshuffle
\eqref{2.5} and arrive at the asymptotic relations
\begin{subequations}\label{2.6}
\begin{align}
 F_l(x,\zl)&=e^{i\zl x}A_l(\zl)+e^{-i\zl x}B_l(\zl)+o(1),&\qquad&x\to-\infty,
\label{2.6a}\\
 F_r(x,\zl)&=e^{-i\zl x}A_r(\zl)+e^{i\zl x}B_r(\zl)+o(1),&\qquad&x\to+\infty,
\label{2.6b}
\end{align}
\end{subequations}
where
\begin{subequations}\label{2.7}
\begin{align}
A_{r,l}(\zl)&=I_2-\frac{1}{2i\zl}\int_{-\infty}^\infty dy\,\bQ(y)m_{r,l}(y,\zl),
\label{2.7a}\\
B_{r,l}(\zl)&=\frac{1}{2i\zl}\int_{-\infty}^\infty dy\,
e^{\mp2i\zl y}\bQ(y)m_{r,l}(y,\zl).\label{2.7b}
\end{align}
\end{subequations}
Then $A_{r,l}(\zl)$ is continuous in $0\neq\zl\in\C^+\cup\R$, is analytic in
$\zl\in\C^+$, and tends to $I_2$ as $\zl\to\infty$ from within $\C^+\cup\R$,
while $2i\zl[I_2-A_{r,l}(\zl)]$ has the finite limit
$-\bzD_{r,l}=\int_{-\infty}^\infty dy\,\bQ(y)m_{r,l}(y,\zl)$ as $\zl\to0$ from
within $\C^+\cup\R$. By the same token, $B_{r,l}(\zl)$ is continuous in
$0\neq\zl\in\R$, vanishes as $\zl\to\pm\infty$, and satisfies
$2i\zl B_{r,l}(\zl)\to-\bzD_{r,l}$ as $\zl\to0$ along the real $\zl$-axis.

Using the transformation $F(x,\zl)\mapsto F(x,-\zl^*)^*$ in the matrix
Schr\"o\-din\-ger equation \eqref{2.1}, we easily prove the symmetry relations
\begin{equation}\label{2.8}
 F_l(x,\zl)=\zs_2F_l(x,-\zl^*)^*\zs_2,\qquad
 F_r(x,\zl)=\zs_2F_r(x,-\zl^*)^*\zs_2.
\end{equation}
With the help of \eqref{2.6} we then obtain the symmetry relations
\begin{subequations}\label{2.9}
\begin{alignat}{4}
A_l(\zl)&=\zs_2A_l(-\zl^*)^*\zs_2,&\ A_r(\zl)&=\zs_2A_r(-\zl^*)^*\zs_2,
&\quad&0\neq\zl\in\C^+\cup\R,\label{2.9a}\\
B_l(\zl)&=\zs_2B_l(-\zl)^*\zs_2,&\ B_r(\zl)&=\zs_2B_r(-\zl)^*\zs_2,
&\quad&0\neq\zl\in\R.\label{2.9b}
\end{alignat}
\end{subequations}

Introducing the {\it reflection coefficients}
\begin{equation}\label{2.10}
R_{l,r}(\zl)=B_{l,r}(\zl)A_{l,r}(\zl)^{-1}=-A_{r,l}(\zl)^{-1}B_{r,l}(-\zl),
\end{equation}
we easily obtain the symmetry relations
\begin{equation}\label{2.11}
R_l(\zl)=\zs_2R_l(-\zl)^*\zs_2,\qquad R_r(\zl)=\zs_2R_r(-\zl)^*\zs_2,\quad
0\neq\zl\in\R,
\end{equation}
provided $\det A_{l,r}(\zl)\neq0$.

Above we have defined $\bzD_{l,r}$ as follows:
$$\bzD_{l,r}=\lim_{\zl\to0}\,2i\zl A_{l,r}(\zl)
=\lim_{\zl\to0^\pm}\,2i\zl B_{l,r}(\zl),$$
where the first limit may be taken from the closed upper half-plane. Then the
matrices $\bzD_{l,r}$ have the same determinant. If $\bzD_{l,r}$ is nonsingular,
we are said to be in the {\it generic case}; if instead $\bzD_{l,r}$ is
singular, we are said to be in the {\it exceptional case} (cf. \cite{AKV}). We
are said to be in the {\it superexceptional case} if $\bzD_{l,r}=0_{2\times2}$
and $A_{l,r}(\zl)$ tends to a nonsingular matrix, $A_{l,r}(0)$ say, as $\zl\to0$
from within $\C^+\cup\R$. It is clear that $\bzD_{l,r}=\zs_2\bzD_{l,r}^*\zs_2$.
Throughout this article (as well as in \cite{DM22}) we assume the absence of
{\it spectral singularities}, i.e., the absence of nonzero real $\zl$ for which
$\det A_{l,r}(\zl)=0$. Under this condition the reflection coefficients
$R_{l,r}(\zl)$ are continuous in $0\neq\zl\in\R$. For general potentials $\bQ$
satisfying \eqref{2.2} or \eqref{2.3} there may very well be spectral
singularities (see \cite{KM,BM} for focusing AKNS examples), even though
spectral singularities do not occur if $\bQ^\dagger=\bQ$ \cite{KM,AW20,DM22}.

The Jost solutions allow the triangular representations
\begin{subequations}\label{2.12}
\begin{align}
 F_l(x,\zl)&=e^{i\zl x}I_2+\int_x^\infty dy\,e^{i\zl y}K(x,y),\label{2.12a}\\
 F_r(x,\zl)&=e^{-i\zl x}I_2+\int_{-\infty}^x dy\,e^{-i\zl y}J(x,y),\label{2.12b}
\end{align}
\end{subequations}
where for every $x\in\R$
$$\int_x^\infty dy\,\|K(x,y)\|+\int_{-\infty}^x dy\,\|J(x,y)\|<+\infty.$$
Then the potential $\bQ(x)$ can be found from the auxiliary functions $K(x,y)$
and $J(x,y)$ as follows:
\begin{equation}\label{2.13}
K(x,x)=\frac{1}{2}\int_x^\infty dy\,\bQ(y),\qquad
J(x,x)=\frac{1}{2}\int_{-\infty}^x dy\,\bQ(y).
\end{equation}
Equations \eqref{2.8} and \eqref{2.12} imply the symmetry relations
\begin{equation}\label{2.14}
K(x,y)=\zs_2K(x,y)^*\zs_2,\qquad J(x,y)=\zs_2J(x,y)^*\zs_2.
\end{equation}
Thus the auxiliary functions $K(x,y)$ and $J(x,y)$ belong to the algebra $\bzS$.

Let us write the reflection coefficients in the form
\begin{equation}\label{2.15}
R_l(\zl)=\int_{-\infty}^\infty d\za\,e^{i\zl\za}\hR_l(\za),\qquad
R_r(\zl)=\int_{-\infty}^\infty d\za\,e^{-i\zl\za}\hR_r(\za),
\end{equation}
where $\hR_{l,r}\in L^1(\R)^{2\times2}$. Although this Fourier representation
has only been proved under the absence of spectral singularities assumption and
in the generic case (for $\bQ\in L^1(\R;(1+|x|)dx)^{2\times2}$) and in the
superexceptional case (for $\bQ\in L^1(\R;(1+|x|)^2dx)^{2\times2}$) \cite{DM21},
we assume it to be also true in the most general exceptional case. We then
easily prove the symmetry relations
\begin{equation}\label{2.16}
\hR_l(\za)=\zs_2\hR_l(\za)^*\zs_2,\qquad\hR_r(\za)=\zs_2\hR_r(\za)^*\zs_2.
\end{equation}
Thus the functions $\hR_l(\za)$ and $\hR_r(\za)$ belong to the algebra $\bzS$.

So far we have only discussed the direct scattering problem for \eqref{2.1}. The
inverse scattering problem can be solved by computing one of the auxiliary
functions $K(x,y)$ or $J(x,y)$ as the solutions of one of the Marchenko integral
equations
\begin{subequations}\label{2.17}
\begin{align}
K(x,y)+\zO_r(x+y)&+\int_x^\infty dz\,K(x,z)\zO_r(z+y)=0_{2\times2},
\label{2.17a}\\
J(x,y)+\zO_l(x+y)&+\int_{-\infty}^x dz\,J(x,z)\zO_l(z+y)=0_{2\times2},
\label{2.17b}
\end{align}
\end{subequations}
followed by an application of one of \eqref{2.13}. Here the {\it Marchenko
integral kernels} $\zO_{l,r}(w)$ are given by
\begin{subequations}\label{2.18}
\begin{align}
\zO_r(w)&=\hR_r(w)+\sum_{s=1}^N\,e^{i\zl_sw}N_{r;s},\label{2.18a}\\
\zO_l(w)&=\hR_l(w)+\sum_{s=1}^N\,e^{-i\zl_sw}N_{l;s},\label{2.18b}
\end{align}
\end{subequations}
where we assume the poles $\zl_s$ ($s=1,\ldots,N$) of the {\it transmission
coefficients} $A_{l,r}(\zl)^{-1}$ to be simple; in that case the so-called
{\it norming constants} $N_{r;s}$ and $N_{l;s}$  are defined by
\begin{subequations}\label{2.19}
\begin{align}
 F_r(x,\zl_s)\tau_{r;s}&=iF_l(x,\zl_s)N_{r;s},\label{2.19a}\\
 F_l(x,\zl_s)\tau_{l;s}&=iF_r(x,\zl_s)N_{l;s},\label{2.19b}
\end{align}
\end{subequations}
where $\tau_{r;s}$ and $\tau_{l;s}$ are the residues of $A_r(\zl)^{-1}$ and
$A_l(\zl)^{-1}$ at the simple pole $\zl_s\in\C^+$ ($s=1,\ldots,N$). If there
exist multiple poles of $A_{l,r}(\zl)^{-1}$ in $\C^+$, then the expressions for
$\zO_{r,l}(w)-\hR_{r,l}(w)$ can be derived in a straightforward way as a finite
sum of polynomials times exponentials which obviously are entire analytic
functions of $x$. We can then prove the symmetry relations
\begin{equation}\label{2.20}
\zO_r(w)=\zs_2\zO_r(w)^*\zs_2,\qquad\zO_l(w)=\zs_2\zO_l(w)^*\zs_2.
\end{equation}
Thus the Marchenko kernels $\zO_r(w)$ and $\zO_l(w)$ belong to the algebra
$\bzS$. The proof can be based on (a)~the unique solvability of the Marchenko
equations (for $\zO_{r,l}$ as unknowns with the auxiliary functions assumed to
be known) for large enough $\pm x$, (b)~the symmetry relations \eqref{2.16}, and
(c)~the analyticity of the functions $\zO_{r,l}(w)-\hR_{r,l}(w)$ in $x\in\R$.
We refer to \cite{DM1} for the rather technical details.

\section{Time evolution}\label{sec:3}

Straightforward calculations imply \cite{DM22}
\begin{align}
i\zs_3\bQ_t+\bQ_{xx}-2\CQ\bQ_x&-2\CQ_x\bQ=(i\zs_3\CQ_t+\CQ_{xx}-2\CQ^3
-2\mu^2\CQ)\CQ\nonumber\\&-\CQ(i\zs_3\CQ_t+\CQ_{xx}-2\CQ^3-2\mu^2\CQ)
\nonumber\\&+(i\zs_3\CQ_t+\CQ_{xx}-2\CQ^3-2\mu^2\CQ)_x.\label{6.1}
\end{align}
Thus any solution of the matrix NLS-like equation \eqref{1.3} with nonvanishing
time invariant limits $\CQ_{r,l}$ for $\CQ(x;t)$ as $x\to\pm\infty$ is
a solution of the nonlinear evolution equation
\begin{equation}\label{3.2}
i\zs_3\bQ_t+\bQ_{xx}-2\CQ\bQ_x-2\CQ_x\bQ=0_{2\times2},
\end{equation}
where $\CQ_x=\tfrac{1}{2}(\bQ-\zs_3\bQ\zs_3)$.

The pair of $4\times4$ matrices $(\bX,\bT)$, where
\begin{subequations}\label{3.3}
\begin{align}
\bX(x,t,\zl)&=\begin{pmatrix}0_{2\times2}&I_2\\ \bQ(x;t)-\zl^2I_2&0_{2\times2}
\end{pmatrix},\label{3.3a}\\
\bT(x,t,\zl)&=\begin{pmatrix}i\zs_3(\bQ-2\zl^2I_2)&-2i\zs_3\CQ\\
i\zs_3(\bQ_x-2\CQ\bQ+2\zl^2\CQ)&i\zs_3(\bQ-2\zl^2I_2-2\CQ_x)\end{pmatrix},
\label{3.3b}
\end{align}
\end{subequations}
is an AKNS pair for the nonlinear evolution equation \eqref{3.2} in the sense
that the zero curvature condition
$$\bX_t-\bT_x+\bX\bT-\bT\bX=0_{4\times4}$$
is satisfied iff $\bQ$ satisfies \eqref{3.2} (see \cite{DM22}). Then it is
easily verified that $\bT(x,t,\zl)$ tends to the limits
\begin{equation}\label{3.4}
\bT_{\pm\infty}=\begin{pmatrix}-2i\zl^2\zs_3&-2i\zs_3\CQ_{r,l}\\
2i\zs_3\CQ_{r,l}&-2i\zl^2\zs_3\end{pmatrix}
\end{equation}
as $x\to\pm\infty$. Note that $\det\bT_{\pm\infty}=16(\zl^4+\mu^2)^2$.

 Following \cite{DM22}, we introduce the Jost solutions $\bF_{r,l}(x,\zl;t)$ of
the first order system
$$\begin{pmatrix}V\\V^\prime\end{pmatrix}^\prime=\begin{pmatrix}0_{n\times n}
&I_n\\ \bQ(x)-\zl^2I_n&0_{n\times n}\end{pmatrix}
\begin{pmatrix}V\\V^\prime\end{pmatrix}$$
defined by
$$\bF_l(x,\zl)=\begin{pmatrix}F_l(x,-\zl)&F_l(x,\zl)\\
 F_l^\prime(x,-\zl)&F_l^\prime(x,\zl)\end{pmatrix},\ \bF_r(x,\zl)
=\begin{pmatrix}F_r(x,\zl)&F_r(x,-\zl)\\F_r^\prime(x,\zl)&F_r^\prime(x,-\zl)
\end{pmatrix},$$
where the prime denotes differentiation with respect to $x$. Letting
$\bV(x,\zl;t)$ be a nonsingular $4\times4$ matrix solution of the pair of
first order equations
\begin{equation}\label{3.5}
\bV_x=\bX\bV,\qquad\bV_t=\bT\bV,
\end{equation}
the fact that $\bF_{r,l}(x,\zl;t)$ satisfies the first of \eqref{3.5} implies
the existence of nonsingular matrices $C_{\bF_{r,l}}(\zl;t)$ not depending on
$x$ such that
$$\bF_{r,l}(x,\zl;t)=\bV(x,\zl;t)C_{\bF_{r,l}}(\zl;t)^{-1}.$$ Then a simple
differentiation yields
$$\left[C_{\bF_{r,l}}(\zl;t)\right]_tC_{\bF_{r,l}}(\zl;t)^{-1}
=\bF_{r,l}^{-1}\bT\bF_{r,l}-\bF_{r,l}^{-1}[\bF_{r,l}]_t,$$
where the left-hand side does not depend on $x$ and hence equals the limits of
the right-hand side as $x\to\pm\infty$. Using \eqref{3.4} we easily get
\begin{equation}\label{3.6}
\left[C_{\bF_{r,l}}(\zl;t)\right]_tC_{\bF_{r,l}}(\zl;t)^{-1}
=\begin{pmatrix}-\zL_{r,l}^\up(\zl)&0_{2\times2}\\
0_{2\times2}&-\zL_{r,l}^\dn(\zl)\end{pmatrix},
\end{equation}
where
\begin{subequations}\label{3.7}
\begin{align}
\zL_{r,l}^\up(\zl)&=2i\zl^2\zs_3+2\zl\zs_3\CQ_{r,l},\label{3.7a}\\
\zL_{r,l}^\dn(\zl)&=2i\zl^2\zs_3-2\zl\zs_3\CQ_{r,l},\label{3.7b}
\end{align}
\end{subequations}
are time invariant. Then we easily verify the symmetry relations
\begin{equation}\label{3.8}
\zL_{r,l}^\up(\zl)=\zs_2\zL_{r,l}^\up(-\zl^*)^*\zs_2,\qquad
\zL_{r,l}^\dn(\zl)=\zs_2\zL_{r,l}^\dn(-\zl^*)^*\zs_2.
\end{equation}

Using that
$$\bF_r(x,\zl;t)=\bF_l(x,\zl;t)\bA_r(\zl;t),\qquad
\bF_l(x,\zl;t)=\bF_r(x,\zl;t)\bA_l(\zl;t),$$
where
$$\bA_r(\zl;t)=\begin{pmatrix}A_r(\zl;t)&B_r(-\zl;t)\\B_r(\zl;t)&A_r(-\zl;t)
\end{pmatrix},\ \bA_l(\zl;t)=\begin{pmatrix}A_l(-\zl;t)&B_l(\zl;t)\\B_l(-\zl;t)
&A_l(\zl;t)\end{pmatrix},$$
for $0\neq\zl\in\R$ we easily compute
\begin{equation}\label{3.9}
[\bA_{r,l}]_t=\bA_{r,l}(\zl;t)\begin{pmatrix}\zL^\up_{r,l}(\zl)&0_{2\times2}\\
0_{2\times2}&\zL^\dn_{r,l}(\zl)\end{pmatrix}-\begin{pmatrix}\zL^\up_{l,r}(\zl)
&0_{2\times2}\\0_{2\times2}&\zL^\dn_{l,r}(\zl)\end{pmatrix}\bA_{r,l}(\zl;t).
\end{equation}
Then the reflection coefficients satisfy
\begin{subequations}\label{3.10}
\begin{align}
[R_r]_t&=R_r(\zl;t)\zL^\up_l(\zl)-\zL^\dn_l(\zl)R_r(\zl;t),\label{3.10a}\\
[R_l]_t&=R_l(\zl;t)\zL^\dn_r(\zl)-\zL^\up_r(\zl)R_l(\zl;t).\label{3.10b}
\end{align}
\end{subequations}
Defining $\hR_{r,l}(\za;t)$ by \eqref{2.15}, we easily derive the PDEs
\begin{subequations}\label{3.11}
\begin{align}
[\hR_r]_t&=-2i\left([\hR_r]_{\za\za}\zs_3-\zs_3[\hR_r]_{\za\za}
+[\hR_r]_\za\zs_3\CQ_l-\CQ_l\zs_3[\hR_r]_\za\right),\label{3.11a}\\
[\hR_l]_t&=-2i\left([\hR_l]_{\za\za}\zs_3-\zs_3[\hR_l]_{\za\za}
+[\hR_l]_\za\zs_3\CQ_r-\CQ_r\zs_3[\hR_l]_\za\right),\label{3.11b}
\end{align}
\end{subequations}
provided $\int_{-\infty}^\infty d\za\,(1+\za^2)\|\hat{R}_{r,l}(\za;t)\|$
converges for every $t\in\R$. Using \eqref{3.11} and time evolution properties
of the norming constants \cite[(4.4)]{DM22} we obtain
\begin{subequations}\label{3.12}
\begin{align}
[\zO_r]_t&=-2i\left([\zO_r]_{ww}\zs_3-\zs_3[\zO_r]_{ww}
+[\zO_r]_w\zs_3\CQ_l-\CQ_l\zs_3[\zO_r]_w\right),\label{3.12a}\\
[\zO_l]_t&=-2i\left([\zO_l]_{ww}\zs_3-\zs_3[\zO_l]_{ww}
+[\zO_l]_w\zs_3\CQ_r-\CQ_r\zs_3[\zO_l]_w\right).\label{3.12b}
\end{align}
\end{subequations}
Hence, the reflection kernels $\hR_{r,l}(\za;t)$ and the Marchenko integral
kernels $\zO_{r,l}(w;t)$ satisfy the same PDEs. We have also seen before that
$\hR_{r,l}(\za;t)$ and $\zO_{r,l}(w;t)$ belong to the algebra $\bzS$.

\section{Quaternionic matrix algebra}\label{sec:4}

Let $\bzS$ stand for the (noncommutative) division ring of complex $2\times2$
matrices $S$ satisfying $S^*=\zs_2S\zs_2$. Then it is easily verified \cite{R3}
that $\bzS$ is isomorphic (as a real unital algebra) to the noncommutative
division ring of quaternions $\Hm$ by means of the isomorphism
\begin{align}
S&=\begin{pmatrix}S_1&-S_2^*\\S_2&S_1^*\end{pmatrix}
=(\text{Re}\,S_1)I_2+i(\text{Im}\,S_1)\zs_3-i(\text{Re}\,S_2)\zs_2
+i(\text{Im}\,S_2)\zs_1\nonumber\\
&\stackrel{\zf}{\longrightarrow}
(\text{Re}\,S_1)\bone+(\text{Im}\,S_1)\bi
-(\text{Re}\,S_2)\bj+(\text{Im}\,S_2)\bk,\label{4.1}
\end{align}
where $\{\bone,\bi,\bj,\bk\}$ is the standard quaternion basis. Thus, letting
$\zs_1=\left(\begin{smallmatrix}0&1\\1&0\end{smallmatrix}\right)$ stand for the
first Pauli matrix, we see that $\{I_2,i\zs_3,i\zs_2,i\zs_1\}$ is the basis of
the real vector space $\bzS$ that corresponds to the quaternion basis
$\{\bone,\bi,\bj,\bk\}$ by means of $\zf$. If
$x=a\bone+b\bi+c\bj+d\bk\in\Hm$ for $a,b,c,d\in\R$, then the quaternion squared
length is defined by $|x|^2=a^2+b^2+c^2+d^2$. Thus, for each
$S=\left(\begin{smallmatrix}S_1&-S_2^*\\S_2&S_1^*\end{smallmatrix}\right)
\in\bzS$ we see that $\det S$ coincides with the squared quaternion length of
$\zf(S)$.

The map $\zf$ has a natural extension as a real algebra isomorphism from
$\bzS^{p\times p}$ onto $\Hm^{p\times p}$, the algebras of $p\times p$ matrices
with entries in $\bzS$ and $\Hm$, respectively. For $p\neq r$ there also exists
a natural extension from the real linear subspace $\bzS^{p\times r}$ onto
$\Hm^{p\times r}$.

 For later use we introduce the {\it similarity orbit} of
$S=\left(\begin{smallmatrix}S_1&-S_2^*\\S_2&S_1^*\end{smallmatrix}\right)
\in\bzS$ as the set \cite[Thm.~2.2.6]{R3}
\begin{align}
&\text{Sim}(S)=\left\{X^{-1}SX:0_{2\times2}\neq X\in\bzS\right\}\label{4.2}\\
&=\left\{\begin{pmatrix}T_1&-T_2^*\\T_2&T_1^*\end{pmatrix}:
\text{Re}\,S_1=\text{Re}\,T_1\ \text{and}\ (\text{Im}\,S_1)^2+|S_2|^2
=(\text{Im}\,T_1)^2+|T_2|^2\right\}.\nonumber
\end{align}

\subsection{Determinants and quaternionic linear algebra}\label{sec:4.1}

Since multiplication of quaternion numbers is noncommutative, there is no
obvious way to define the determinant of square quaternion matrices.
 Fortunately, the map $\zf$ allows one to define the determinant of a quaternion
$p\times p$ matrix $M\in\Hm^{p\times p}$ as the determinant of the complex
$2p\times2p$ matrix $\zf^{-1}(M)$ (cf. \cite[Ch.~5]{R3}). For alternative ways
to define determinants of square quaternionic matrices we refer to
\cite{St,Dd,CL} and references therein.

The following theorem has been proved by Rodman \cite[Th.~5.9.2]{R3} using the
quaternionic Jordan normal form. Below we present an independent proof based on
Schur complements (cf. \cite[Sec.~1.7]{Dym} and references therein).

\begin{theorem}\label{th:4.1}
 For $p=1,2,3,\ldots$ the matrices $S\in\bzS^{p\times p}$ have a nonnegative
determinant.
\end{theorem}

\begin{proof}
 For $p=1$ the theorem is obviously true. For $p\ge2$ we define the Schur
complement
\begin{equation}\label{4.3}
\CS=\begin{pmatrix}S_{22}&\ldots&S_{2p}\\ \vdots&&\vdots\\
S_{p2}&\ldots&S_{pp}\end{pmatrix}-\begin{pmatrix}S_{21}\\ \vdots\\S_{p1}
\end{pmatrix}S_{11}^{-1}\begin{pmatrix}S_{12}&\ldots&S_{1p}\end{pmatrix},
\end{equation}
provided $\det(S_{11})=\|S_{11}\|^2>0$. Then
\begin{equation}\label{4.4}
\det(S)=\|S_{11}\|^2\det\CS.
\end{equation}
Under the induction hypothesis that all matrices $S\in\bzS^{(p-1)\times(p-1)}$
have a nonnegative determinant, we see from \eqref{4.4} that any matrix
$S\in\bzS^{p\times p}$ satisfying $\|S_{11}\|>0$ has a nonnegative determinant.
If one of $\|S_{j1}\|>0$, we switch the first and $j$-th double rows without
changing the determinant and repeat the above Schur complement argument to
conclude that $\det(S)\ge0$. If $\|S_{11}\|=\ldots=\|S_{p1}\|=0$, then obviously
$\det(S)=0$.
\end{proof}

\subsection{Jordan normal form and matrix triplets}\label{sec:4.2}

The following theorem can be obtained from \cite[Thm.~5.5.3]{R3} upon
application of $\zf^{-1}$.

\begin{theorem}\label{th:4.2}
 For every $S\in\bzS^{p\times p}$ there exist positive integers $m_1,\ldots,m_k$
adding up to $p$ and matrices $A^{[1]},\ldots,A^{[k]}\in\bzS$ such that $S$ is
similar to the direct sum
\begin{equation}\label{4.5}
J_{m_1}(A^{[1]})\oplus\ldots\oplus J_{m_k}(A^{[k]})
\end{equation}
by means of a similarity transformation belonging to $\bzS^{p\times p}$. The
$\bzS$-Jordan normal form \eqref{4.5} is unique up to changing the order in the
direct sum and replacing the matrices $A^{[1]},\ldots,A^{[k]}$ by matrices in
the same similarity orbit.
\end{theorem}

It should be noted that the $\bzS$-Jordan normal form (or: the quaternionic
Jordan normal form discussed at length in \cite{R3}) differs from the usual
complex Jordan normal form. Since
$A=\left(\begin{smallmatrix}A_1&-A_2^*\\A_2&A_1^*\end{smallmatrix}\right)$ is
a diagonalizable $2\times2$ matrix with eigenvalues
$\text{Re}\,A_1\pm i\sqrt{(\text{Im}\,A_1)^2+|A_2|^2}$, the corresponding
complex Jordan normal form is obtained from \eqref{4.5} below as follows:
\begin{itemize}
\item[1.] If $A^{[s]}$ is the diagonal matrix
$\left(\text{Re}\,A^{[s]}_1\right)I_2$, we replace $J_{m_s}(A^{[s]})$ by the
direct sum of Jordan blocks
$J_{m_s}(\text{Re}\,A_1)\oplus J_{m_s}(\text{Re}\,A_1)$.
\item[2.] If $A^{[s]}$ is not a real multiple of $I_2$, we replace
$J_{m_s}(A^{[s]})$ by the direct sum of the Jordan blocks of order $m_s$ at the
complex conjugate eigenvalues
$\text{Re}\,A_1\pm i\sqrt{(\text{Im}\,A_1)^2+|A_2|^2}$.
\end{itemize}

Let $(\bA,\bB,\bC)$ be a triplet consisting of the $p\times p$ matrix $\bA$ with
entries in $\bzS$, the $p\times1$ matrix $\bB$ with entries in $\bzS$, and the
$1\times p$ matrix $\bC$ with entries in $\bzS$. Then this matrix triplet is
called {\it minimal} if the matrix order of $\bA$ is minimal among all triplets
for which $\bC e^{-z\bA}\bB$ is the same $\bzS$-valued function of $z\in\R$.
According to Theorem \ref{th:4.2}, given a minimal triple $(\bA,\bB,\bC)$ of
matrices with entries in $\bzS$ there exists an invertible
$\bS\in\bzS^{p\times p}$ such that $\bS\bA\bS^{-1}$ has the Jordan normal form
\eqref{4.5} and the triplet $(\bS\bA\bS^{-1},\bS\bB,\bC\bS^{-1})$ is minimal.

\begin{theorem}\label{th:4.3}
Suppose $(\bA,\bB,\bC)$ is a triplet of size compatible matrices with entries of
$\bzS$, where the eigenvalues of $\bA$ all have positive real part. Let us
assume that $\bA$ has been brought to the $\bzS$-Jordan normal form
\eqref{4.5}. Then no pair of matrices $A^{[1]},\ldots,A^{[k]}$ belongs to the
same similarity orbit and among the $\bzS$-entries $B^{[j]}_s$ of $\bB$ and
$C^{[j]}_s$ of $\bC$ ($j=1,\ldots,k$, $s=1,\ldots,m_j$ the $\bzS$-entries
$B^{[j]}_{m_j}$ and $C^{[j]}_1$ ($j=1,2,\ldots,k$) are nontrivial matrices.
\end{theorem}

\begin{proof}
Consider the matrix triplet $(J_m(A),\bB,\bC)$, where $A\in\bzS$ is not the zero
matrix, $\bB$ is the column with entries $B_1,\ldots,B_m\in\bzS$ and $\bC$ is
the row with entries $C_1,\ldots,C_m\in\bzS$. Then for $n=0,1,2,\ldots$ we get
$$\left[J_m(A)^n\right]_{j,l}=\begin{cases}A^n,&j=l,\\
\binom{n}{l-j}A^{n+j-l},&j<l\le\min(m-1,n+j),\\0_{2\times2},&j>l
\ \text{or}\ l>\min(m-1,n+j),\end{cases}$$
which is an upper triangular Toeplitz matrix with entries in $\bzS$. Letting
$\bX$ be the column with entries $X_1,\ldots,X_m\in\bzS$, the identity
$$\bC J_m(A)^n\bX=0_{2\times2},\qquad n=0,1,\ldots,m-1,$$
allows a solution $\bX$ with $X_1\neq0_{2\times2}$ if $C_1=0_{2\times2}$. Thus
assuming $C_1\neq0_{2\times2}$ in $\bzS$, we get the equality
$$\sum_{j=1}^m\,C_j\left(A^nX_j+\sum_{l=j+1}^{\min(m-1,n+1)}\binom{n}{l-j}
A^{n+j-l}X_l\right)=0_{2\times2},\quad n=0,1,2,\ldots,$$
allowing us to express each $X_j$ into $X_{j+1},\ldots,X_m$ ($j=1,2,\ldots,m-1$)
linearly and to conclude that $X_m=0_{2\times2}$. Thus
$X_1=\ldots=X_m=0_{2\times2}$. In other words, if $C_1\neq0_{2\times2}$, then
$$\bigcap_{n=0}^{m-1}\,\text{Ker}\,(\bC J_m(A)^n)=(0_{2\times2}).$$

In the same way we prove that
$$\bigvee_{n=0}^{m-1}\,\text{Im}\,(J_m(A)^n\bB)=\C^{2m}$$
if $B_m\neq0_{2\times2}$.
\end{proof}

\section{Soliton solutions using matrix triplets}\label{sec:5}

Let us now solve the right and left Marchenko equations \eqref{2.17a} and
\eqref{2.17b} for reflectionless Marchenko kernels \eqref{2.18a} and
\eqref{2.18b}, where the reflection coefficients $R_{r,l}(\zl;t)$ vanish.

\subsection{Minimal matrix triplet representations}\label{sec:5.1}

Since the Marchenko kernels $\zO_{l,r}(w;t)$ are finite linear combinations of
the exponentials $e^{\pm i\zl_sw}$ ($n=1,2,\ldots,N$) and polynomials of $w$
multiplied by such exponentials with time dependent coefficients, there exist
a square matrix $\bA$ of even order $2p$ whose eigenvalues have positive real
parts, $2p\times2$ matrices $\bB_r$ and $\bB_l$, $2\times2p$ matrices $\bC_r$
and $\bC_l$, and a $2p\times2p$ matrix $\bH$ commuting with $\bA$ such that
\begin{equation}\label{5.1}
\zO_r(z,t)=\bC_re^{-z\bA}e^{t\bH}\bB_r,\qquad
\zO_l(z,t)=\bC_le^{z\bA}e^{t\bH}\bB_l.
\end{equation}
The representations \eqref{5.1} are chosen in such a way that the order of the
\text{complex} matrix $\bA$ is minimal among all representations \eqref{5.1} for
the same Marchenko kernels $\zO_r(z,t)$ and $\zO_l(z,t)$. In that case $2p$
coincides with the sum of the algebraic multiplicities of the discrete
eigenvalues in $\C^+$ (which is $N$ if the discrete eigenvalues are
algebraically simple, as assumed so far). Moreover, for any pair of minimal
representations \eqref{5.1} [where the matrices in the second pair carry a prime
or double prime, respectively], there exist unique nonsingular $2p\times2p$
complex matrices $\bS$ and $\overline{\bS}$ such that \cite[Ch.~1]{BGK}
\begin{subequations}\label{5.2}
\begin{alignat}{4}
\bA'&=\bS\bA\bS^{-1},&\quad\bB_r'&=\bS\bB_r,&\quad\bC_r'&=\bC_r\bS^{-1},&\quad
\bH'&=\bS\bH\bS^{-1},\label{5.2a}\\
\bA''&=\overline{\bS}\bA\overline{\bS}^{-1},&\quad\bB_l''&=\overline{\bS}\bB_l,
&\quad\bC_l''&=\bC_l\overline{\bS}^{-1},&\quad\bH''&=\overline{\bS}\bH
\overline{\bS}^{-1}.\label{5.2b}
\end{alignat}
\end{subequations}
In other words, choosing the primed and double primed matrix quadruplets to be
$(\bA^*,\bB_{r,l}^*\zs_2,\zs_2\bC_{r,l}^*,\bH^*)$, the symmetry relations
\eqref{2.20} for the Marchenko kernels imply the existence of unique nonsingular
$2p\times2p$ matrices $\bS$ and $\overline{\bS}$ such that
\begin{subequations}\label{5.3}
\begin{alignat}{4}
\bA^*&=\bS\bA\bS^{-1},&\quad\bB_r^*\zs_2&=\bS\bB_r,&\quad
\zs_2\bC_r^*&=\bC_r\bS^{-1},&\quad\bH^*&=\bS\bH\bS^{-1},\label{5.3a}\\
\bA^*&=\overline{\bS}\bA\overline{\bS}^{-1},
&\quad\bB_l^*\zs_2&=\overline{\bS}\bB_l,&\quad
\zs_2\bC_l^*&=\bC_l\overline{\bS}^{-1},&\quad
\bH^*&=\overline{\bS}\bH\overline{\bS}^{-1}.\label{5.3b}
\end{alignat}
\end{subequations}
Taking complex conjugates we get
\begin{subequations}\label{5.4}
\begin{align}
\begin{cases}\bA^*={\bS^*}^{-1}\bA\bS^*,&\bB_r^*\zs_2=-{\bS^*}^{-1}\bB_r,\\
\zs_2\bC_r^*=-\bC_r\bS^*,&\bH^*={\bS^*}^{-1}\bH\bS^*,\end{cases}\label{5.4a}\\
\begin{cases}\bA^*={\overline{\bS}^*}^{-1}\bA\overline{\bS}^*,
&\bB_l^*\zs_2=-{\overline{\bS}^*}^{-1}\bB_l,\\
\zs_2\bC_l^*=-\bC_l\overline{\bS}^*,
&\bH^*={\overline{\bS}^*}^{-1}\bH\overline{\bS}^*.\end{cases}\label{5.4b}
\end{align}
\end{subequations}
The uniqueness of the similarity transformations $\bS$ and $\overline{\bS}$ then
implies that
\begin{equation}\label{5.5}
\bS^*=-\bS^{-1},\qquad\overline{\bS}^*=-\overline{\bS}^{-1}.
\end{equation}
We observe that the minimal matrix triplets $(\bA_r,\bB_r,\bC_r)$
and $(\bA_l,\bB_l,\bC_l)$ need not consist of matrices having their entries in
$\bzS$, even though the expressions $\bC_re^{-z\bA_r}\bB_r$ and
$\bC_le^{z\bA_l}\bB_l$ belong to $\bzS$ for each $z\in\R$.

Let us now apply a similarity transformation to the triplets
$(\bA_r,\bB_r,\bC_r)$ and $(\bA_l,\bB_l,\bC_l)$ such that the newly found
triplets consist of matrices having their entries in $\bzS$. Indeed, letting
$\bT=\zl^{-1}(\zS_2+\zl^2\bS^*)$ where $\zS_2$ is the direct sum of $p$ copies
of $\zs_2$, $|\zl|=1$, and $\zS_2+\zl^2\bS^*$ is nonsingular, we obtain
$$\bS\bT\zS_2=\zl^{-1}\bS+\zl\bS\bS^*\zS_2=\zl^{-1}\bS-\zl\zS_2
=(\zl^{-1}\zS_2+\zl\bS^*)^*=\bT^*,$$
and hence $\bS=\bT^*\zS_2\bT^{-1}$ (see \cite{ML} for a similar argument
involving the Ansatz $\bS^*=\bS^{-1}$). Substituting the latter into
\eqref{5.4a} we get
$$\begin{cases}(\bT^{-1}\bA\bT)^*=\zS_2(\bT^{-1}\bA\bT)\zS_2,\\
(\bT^{-1}\bB)^*=\zS_2(\bT^{-1}\bB)\zs_2,\\
(\bC\bT)^*=\zs_2(\bC\bT)\zS_2,\end{cases}$$
where we have omitted the subscripts $r$ and $l$. Hence, the matrix triplet
$(\bT^{-1}\bA\bT,\bT^{-1}\bB,\bC\bT)$ consists of matrices having their entries
in $\bzS$. In the same way, by replacing $\bH$ with $\bT^{-1}\bH\bT$ we arrive
at a matrix belonging to $\bzS^{p\times p}$.

Since the Zakharov-Shabat system $v_x=(-ik\zs_3+\CQ)v$ is $1+1$, every discrete
eigenvalue $k_s\in\C^+$ is {\it geometrically} simple. Because the conformal
mapping $k\mapsto\zl=\sqrt{k^2+\mu^2}$ is $1,1$ on $\C^+$ cut along the segment
$(i0^+,i\mu]$, the eigenvalues $\zl_s$ of the matrix Schr\"o\-din\-ger equation
\eqref{2.1} in $\C^+$ are geometrically simple. Thus the matrix
$\bA_{r,l}$ in the minimal representations \eqref{5.1} has a $\bzS$-Jordan
structure with exactly two Jordan blocks of the same order per positive
eigenvalue, one Jordan block per complex eigenvalue with positive real part, and
Jordan blocks of the same order corresponding to complex conjugate eigenvalues
(which have positive real part). As a result, there exist quadruplets
$(\bA_r,\bB_r,\bC_r,\bH_r)$ and $(\bA_l,\bB_l,\bC_l,\bH_l)$ consisting of
matrices having their entries in $\bzS$ such that $\bA_r$ and $\bA_l$ have the
above $\bzS$-Jordan normal form and have minimal matrix order among all
quadruplets leading to the same Marchenko integral kernels \eqref{5.1}.

\subsection{Inverse scattering implemented}\label{sec:5.2}

Let us depart from the representations \eqref{5.1} of the Marchenko integral
kernels, where the quadruplets $(\bA_r,\bB_r,\bC_r,\bH_r)$ and
$(\bA_l,\bB_l,\bC_l,\bH_l)$ consist of matrices having their entries in
$\bzS$ such that $\bA_r$ and $\bA_l$ have the above $\bzS$-Jordan normal form
and have minimal matrix order among all quadruplets leading to the same
Marchenko integral kernels \eqref{5.1}.

Substituting the first of \eqref{5.1} into the right Marchenko equation
\eqref{2.17a}, we obtain using the commutativity of $\bA$ and $\bH$
\begin{align}
K(x,y;t)&=-\left[\bC_re^{-x\bA}+\int_x^\infty ds\,K(x,s;t)\bC_re^{-s\bA}
\right]e^{-y\bA}e^{t\bH}\bB_r\nonumber\\&=-\bW_r(x;t)e^{-y\bA}e^{t\bH}\bB_r,
\label{5.6}
\end{align}
where $\bW_r(x;t)=\bC_re^{-x\bA}-\bW_r(x;t)e^{-x\bA}e^{t\bH}\bP_re^{-x\bA}$ and
\begin{equation}\label{5.7}
\bP_r=\int_0^\infty ds\,e^{-s\bA}\bB_r\bC_re^{-s\bA}
\end{equation}
is the unique solution of the Sylvester equation $\bA\bP_r+\bP_r\bA=\bB_r\bC_r$.
Hence,
\begin{equation}\label{5.8}
K(x,y;t)=-\bC_re^{-x\bA}\left[I_{2p}+e^{-x\bA}e^{t\bH}\bP_re^{-x\bA}
\right]^{-1}e^{-y\bA}e^{t\bH}\bB_r,
\end{equation}
provided the inverse matrix exists. Then Theorem \ref{th:A.1} implies that
$\bP_r$ is invertible. Moreover, Theorem \ref{th:A.3} implies that the inverse
matrix in \eqref{5.8} exists for all but finitely many $x\in\R$. Similarly,
substituting the second of \eqref{5.1} into the left Marchenko equation
\eqref{2.17b}, we obtain
\begin{align}
J(x,y;t)&=-\left[\bC_le^{x\bA}+\int_{-\infty}^x ds\,J(x,s;t)\bC_le^{s\bA}
\right]e^{y\bA}e^{t\bH}\bB_l\nonumber\\&=-\bW_l(x;t)e^{y\bA}e^{t\bH}\bB_l,
\label{5.9}
\end{align}
where $\bW_l(x;t)=\bC_le^{x\bA}-\bW_l(x;t)e^{x\bA}e^{t\bH}\bP_le^{x\bA}$ and
\begin{equation}\label{5.10}
\bP_l=\int_0^\infty ds\,e^{-s\bA}\bB_l\bC_le^{-s\bA}
\end{equation}
is the unique solution of the Sylvester equation $\bA\bP_l+\bP_l\bA=\bB_l\bC_l$.
Then Theorem \ref{th:A.1} implies that $\bP_l$ is invertible. Analogously,
\begin{equation}\label{5.11}
J(x,y;t)=-\bC_le^{x\bA}\left[I_{2p}+e^{x\bA}e^{t\bH}\bP_le^{x\bA}\right]^{-1}
e^{y\bA}e^{t\bH}\bB_l,
\end{equation}
provided the inverse matrix exists. Moreover, Theorem \ref{th:A.3} implies that
the inverse matrix in \eqref{5.11} exists for all but finitely many $x\in\R$.
 Furthermore, $\bP_r$ and $\bP_l$ belong to $\bzS^{p\times p}$.

Using \eqref{2.13} in  \eqref{5.8} and \eqref{5.11} and differentiating with
respect to $x$ we obtain
\begin{subequations}\label{5.12}
\begin{align}
\bQ(x;t)&=-4\bC_r\left[e^{2x\bA}e^{-t\bH}+\bP_r\right]^{-1}\bA e^{2x\bA}
e^{-t\bH}\left[e^{2x\bA}e^{-t\bH}+\bP_r\right]^{-1}\bB_r,\label{5.12a}\\
\bQ(x;t)&=-4\bC_l\left[e^{-2x\bA}e^{-t\bH}+\bP_l\right]^{-1}\bA e^{-2x\bA}
e^{-t\bH}\left[e^{-2x\bA}e^{-t\bH}+\bP_l\right]^{-1}\bB_l.\label{5.12b}
\end{align}
\end{subequations}
Since $\bP_r$ and $\bP_l$ are nonsingular, these expressions are exponentially
decaying as $x\to\pm\infty$. Writing
$\bB_r=\begin{pmatrix}\bB_{r,1}&\bB_{r,2}\end{pmatrix}$ and
$\bC_r=\left(\begin{smallmatrix}\bC_{r,1}\\ \bC_{r,2}\end{smallmatrix}\right)$
and similarly for $\bB_l$ and $\bC_l$, we obtain the following expressions
relating the potentials to the asymptotic potentials $q_r$ and $q_l$ 
\begin{subequations}\label{5.13}
\begin{align}
q(x,t)&=q_r+2\bC_{r,1}\left[e^{2x\bA}e^{-t\bH}+\bP_r\right]^{-1}\bB_{r,2}
\nonumber\\
&=q_l-2\bC_{l,1}\left[e^{-2x\bA}e^{-t\bH}+\bP_l\right]^{-1}\bB_{l,2},
\label{5.13a}\\
q^*(x,t)&=q_r^*+2\bC_{r,2}\left[e^{2x\bA}e^{-t\bH}+\bP_r\right]^{-1}\bB_{r,1}
\nonumber\\
&=q_l^*-2\bC_{l,2}\left[e^{-2x\bA}e^{-t\bH}+\bP_l\right]^{-1}\bB_{r,1},
\label{5.13b}
\end{align}
\end{subequations}
provided $e^{\pm2x\bA}e^{t\bH}+\bP_{r,l}$ (for each $x\in\R$) are nonsingular
matrices. Since $\bP_{r,l}$ are nonsingular, we get
\begin{subequations}\label{5.14}
\begin{alignat}{4}
q_l&=q_r+2\bC_{r,1}\bP_r^{-1}\bB_{r,2},&\qquad
q_r&=q_l-2\bC_{l,1}\bP_l^{-1}\bB_{l,2},\label{5.14a}\\
q_l^*&=q_r^*+2\bC_{r,2}\bP_r^{-1}\bB_{r,1},&\qquad
q_r^*&=q_l^*-2\bC_{l,2}\bP_l^{-1}\bB_{r,1}.\label{5.14b}
\end{alignat}
\end{subequations}
Since $\mu=|q_l|=|q_r|$, the right and left matrix triplets cannot be chosen
arbitrarily. The first of \eqref{5.14a} implies that
$$\bC_{r,1}\bP_r^{-1}\bB_{r,2}=\tfrac{1}{2}\mu(e^{i\theta_l}-e^{i\theta_r}).$$
Since $|e^{i\theta_l}-e^{i\theta_r}|\le2$, we see that the matrix triplet is to
satisfy
\begin{equation}\label{5.15}
|\bC_{r,1}\bP_r^{-1}\bB_{r,2}|
=\mu\left|\sin[\tfrac{1}{2}(\theta_r-\theta_l)]\right|\le\mu,
\end{equation}
where $e^{i(\theta_l-\theta_r)}$ and hence $e^{i\theta_l}$ can be evaluated from
known $\mu$ and $e^{i\theta_r}$. This means that the triplet and $\mu$ are not
independent. Once $\mu$ has been chosen to satisfy
$\mu\ge|\bC_{r,1}\bP_r^{-1}\bB_{r,2}|$, it is possible to determine $\theta_l$
uniquely up to an additive multiple of $2\pi$. Moreover, we have established the following
\begin{proposition}\label{soliton}
If
$0<\mu<|\bC_{r,1}\bP_r^{-1}\bB_{r,2}|$, no soliton solution exists.
\end{proposition}

The matrix $\bH$ commuting with $\bA$ is easily seen to be given by
\begin{equation}\label{5.16}
\bH=\frac{1}{2\pi i}\oint_\zG d\zl\,[2i\zl\sqrt{\zl^2-\mu^2}-i\mu^2]
(\zl I_{2p}-i\bA)^{-1},
\end{equation}
where $k(\zl)=\sqrt{\zl^2-\mu^2}$ is the conformal mapping from $\C^+$ onto
$\C^+$ satisfying $k(\zl)\sim\zl$ at infinity and $\zG$ is a closed rectifiable
Jordan contour in the upper half-plane which has winding number $+1$ with
respect to each eigenvalue of $i\bA$. Then
\begin{equation}\label{5.17}
e^{t\bH}=\frac{1}{2\pi i}\oint_\zG d\zl\,e^{i[2\zl\sqrt{\zl^2-\mu^2}-\mu^2]t}
(\zl I_{2p}-i\bA)^{-1}.
\end{equation}

Let us finally derive the expressions for the transmission coefficients.
Substituting \eqref{5.8} into \eqref{2.12a} and \eqref{5.9} into \eqref{2.12b}
we get
\begin{align*}
 F_l(x,\zl;t)&=e^{i\zl x}\left(I_2-i\bC_r\left[e^{2x\bA}e^{-t\bH}+\bP_r
\right]^{-1}(\zl I_{2p}+i\bA)^{-1}\bB_r\right),\\
 F_r(x,\zl;t)&=e^{-i\zl x}\left(I_2-i\bC_l\left[e^{-2x\bA}e^{t\bH}+\bP_l
\right]^{-1}(\zl I_{2p}+i\bA)^{-1}\bB_l\right).
\end{align*}
Dividing by $e^{\pm i\zl x}$, taking the limits of the resulting equalities as
$x\to\mp\infty$, and using \eqref{2.6a} and \eqref{2.6b} we arrive at the
identities
\begin{subequations}\label{5.18}
\begin{align}
A_l(\zl)&=I_2-i\bC_r\bP_r^{-1}(\zl I_{2p}+i\bA)^{-1}\bB_r,\label{5.18a}\\
A_r(\zl)&=I_2-i\bC_l\bP_l^{-1}(\zl I_{2p}+i\bA)^{-1}\bB_l,\label{5.18b}
\end{align}
\end{subequations}
where we have used the nonsingularity of $\bP_{r,l}$. Using the Sylvester
equations for $\bP_{r,l}$ we obtain the transmission coeffients
\begin{subequations}\label{5.19}
\begin{align}
A_l(\zl)^{-1}&=I_2+i\bC_r(\zl I_{2p}-i\bA)^{-1}\bP_r^{-1}\bB_r,\label{5.19a}\\
A_r(\zl)^{-1}&=I_2+i\bC_l(\zl I_{2p}-i\bA)^{-1}\bP_l^{-1}\bB_l.\label{5.19b}
\end{align}
\end{subequations}
Observe that the transmission coefficients are time-invariant. Using the
Sherman-Morrison-Woodbury formula $\det(I-TS)=\det(I-ST)$ [cf.~\cite{GL}] and
the Sylvester equations for $\bP_{r,l}$ we easily obtain
$$\det[A_{l,r}(\zl)^{-1}]=\frac{\det(\zl I_{2p}+i\bA)}{\det(\zl I_{2p}-i\bA)}.$$

\section{Examples}\label{sec:6}

In this section we work out various examples of multisoliton solutions based on
the minimal quadruplet $(\bA,\bB,\bC,\bH)$, where $\bH=\phi(\bA)$ for some
function $\phi$ that is analytic in a neighborhood of the eigenvalues of $\bA$.
In fact [cf. \eqref{5.17}], $\phi(\zl)=i[2\zl\sqrt{\zl^2-\mu^2}-\mu^2]$, where
$k(\zl)=\sqrt{\zl^2-\mu^2}$ is the conformal mapping from $\C^+$ onto $\C^+$
satisfying $k(\zl)\sim\zl$ at infinity.

\begin{example}[one-soliton solution with real eigenvalue]\label{th:6.1}
Consider the minimal triplet
$$\bA=\begin{pmatrix}a&0\\0&a\end{pmatrix},\qquad
\bB=\begin{pmatrix}b_1&-b_2^*\\b_2&b_1^*\end{pmatrix},\qquad
\bC=\begin{pmatrix}c_1&-c_2^*\\c_2&c_1^*\end{pmatrix},$$
where $a>0$ and $\bB$ and $\bC$ have positive determinants. Then
$$\bP=\frac{1}{2a}\bB\bC
=\frac{1}{2a}\begin{pmatrix}d_1&-d_2^*\\d_2&d_1^*\end{pmatrix},$$
where $d_1=b_1c_1-b_2^*c_2$ and $d_2=b_2c_1+b_1^*c_2$. Then \eqref{5.8} implies
that
\begin{align*}
K(x,y;t)&=-e^{-a(x+y)}e^{t\phi(a)}
\begin{pmatrix}c_1&-c_2^*\\c_2&c_1^*\end{pmatrix}\times\nonumber\\
&\times\begin{pmatrix}1+\tfrac{1}{2a}e^{-2ax}e^{t\phi(a)}d_1
&-\tfrac{1}{2a}e^{-2ax}e^{t\phi(a)}d_2^*\\
\tfrac{1}{2a}e^{-2ax}e^{t\phi(a)}d_2
&1+\tfrac{1}{2a}e^{-2ax}e^{t\phi(a)}d_1^*\end{pmatrix}^{-1}
\begin{pmatrix}b_1&-b_2^*\\b_2&b_1^*\end{pmatrix},
\end{align*}
where for any $x\in\R$ the matrix to be inverted has the nonnegative
determinant
$$D(x;t)=\left|1+\tfrac{1}{2a}e^{-2ax}e^{t\phi(a)}d_1\right|^2
+\left|\tfrac{1}{2a}e^{-2ax}e^{t\phi(a)}d_2\right|^2.$$
We assume this determinant to be positive for each $(x,t)\in\R^2$. In fact,
the determinant $D(x;t)$ vanishes at some $x\in\R$ for given $t\in\R$ [namely,
at $x=\tfrac{1}{2a}\ln(-\tfrac{d_1}{2a}e^{t\phi(a)})$] iff $d_1<0$ and $d_2=0$,
i.e., iff $\bB\bC$ is a negative multiple of $I_2$. Therefore,
\begin{align*}
K(x,y;t)&=-\frac{e^{-a(x+y)}e^{t\phi(a)}}{D(x;t)}
\begin{pmatrix}c_1&-c_2^*\\c_2&c_1^*\end{pmatrix}\times\nonumber\\
&\times\begin{pmatrix}1+\tfrac{1}{2a}e^{-2ax}e^{t\phi(a)}d_1^*
&\tfrac{1}{2a}e^{-2ax}e^{t\phi(a)}d_2^*\\
-\tfrac{1}{2a}e^{-2ax}e^{t\phi(a)}d_2
&1+\tfrac{1}{2a}e^{-2ax}e^{t\phi(a)}d_1\end{pmatrix}
\begin{pmatrix}b_1&-b_2^*\\b_2&b_1^*\end{pmatrix}.
\end{align*}
Consequently,
\begin{align*}
q(x)&=q_r+2\frac{e^{-2ax}e^{t\phi(a)}}{D(x;t)}\begin{pmatrix}c_1&-c_2^*
\end{pmatrix}\times\nonumber\\
&\times\begin{pmatrix}1+\tfrac{1}{2a}e^{-2ax}e^{t\phi(a)}d_1^*
&\tfrac{1}{2a}e^{-2ax}e^{t\phi(a)}d_2^*\\
-\tfrac{1}{2a}e^{-2ax}e^{t\phi(a)}d_2
&1+\tfrac{1}{2a}e^{-2ax}e^{t\phi(a)}d_1\end{pmatrix}
\begin{pmatrix}-b_2^*\\b_1^*\end{pmatrix}\nonumber\\
&=q_r+\frac{2}{D(x;t)}\left[-(c_1b_2^*+c_2^*b_1^*)e^{-2ax}e^{t\phi(a)}\right.\\
&+\left.\frac{1}{2a}e^{-4ax}e^{2t\phi(a)}\left(-c_1d_1^*b_2^*+c_1d_2^*b_1^*
-c_2^*d_2b_2^*-c_2^*d_1b_1^*\right)\right].
\end{align*}
Thus,
\begin{align*}
q_l&=q_r+\frac{4a}{|d_1|^2+|d_2|^2}\begin{pmatrix}c_1&-c_2^*\end{pmatrix}
\begin{pmatrix}d_1^*&d_2^*\\-d_2&d_1\end{pmatrix}\begin{pmatrix}-b_2^*\\b_1^*
\end{pmatrix}=q_r.
\end{align*}
Since $\bP=\tfrac{1}{2a}\bB\bC$ with $\bB$ and $\bC$ nonsingular, we see that
$$q_l-q_r=4a\begin{pmatrix}1&0\end{pmatrix}\bC(\bB\bC)^{-1}\bB\begin{pmatrix}
0\\1\end{pmatrix}=0,$$
thus conferming our preceding result.
\end{example}

\begin{example}[one-soliton solution with conjugate eigenvalues]\label{th:6.2}
Consider the minimal triplet
$$\bA=\begin{pmatrix}a&\zo\\-\zo&a\end{pmatrix},\qquad
\bB=\begin{pmatrix}b_1&-b_2^*\\b_2&b_1^*\end{pmatrix},\qquad
\bC=\begin{pmatrix}c_1&-c_2^*\\c_2&c_1^*\end{pmatrix},$$
where $a>0$, $0\neq\,\zo\in\R$, and $\bB$ and $\bC$ have positive
determinants. Then
$$e^{-x\bA}=e^{-ax}\begin{pmatrix}\cos(\zo x)&\sin(\zo x)\\-\sin(\zo x)
&\cos(\zo x)\end{pmatrix}.$$
Using $\int_0^\infty dx\,e^{-2ax}\cos(2\zo x)=\tfrac{a}{2(a^2+\zo^2)}$ and
$\int_0^\infty dx\,e^{-2ax}\sin(2\zo x)=\tfrac{\zo}{2(a^2+\zo^2)}$, we get the
Sylvester solution
$$\bP=\begin{pmatrix}\tfrac{d_1-d_1^*}{4a}+\tfrac{a(d_1+d_1^*)}{4(a^2+\zo^2)}
+\tfrac{\zo(d_2+d_2^*)}{4(a^2+\zo^2)}&\tfrac{d_2-d_2^*}{4a}
-\tfrac{a(d_2+d_2^*)}{4(a^2+\zo^2)}+\tfrac{\zo(d_1+d_1^*)}{4(a^2+\zo^2)}\\
\tfrac{d_2-d_2^*}{4a}+\tfrac{a(d_2+d_2^*)}{4(a^2+\zo^2)}
-\tfrac{\zo(d_1+d_1^*)}{4(a^2+\zo^2)}&-\tfrac{d_1-d_1^*}{4a}
+\tfrac{a(d_1+d_1^*)}{4(a^2+\zo^2)}+\tfrac{\zo(d_2+d_2^*)}{4(a^2+\zo^2)}
\end{pmatrix},$$
where $d_1=b_1c_1-b_2^*c_2$ and $d_2=b_2c_1+b_1^*c_2$. Note that
\begin{align*}
\det\bP&=\left[\frac{a(d_1+d_1^*)}{4(a^2+\zo^2)}
+\frac{\zo(d_2+d_2^*)}{4(a^2+\zo^2)}\right]^2
+\left[\frac{d_1-d_1^*}{4ia}\right]^2\nonumber\\
&+\left[\frac{a(d_2+d_2^*)}{4(a^2+\zo^2)}-\frac{\zo(d_1+d_1^*)}{4(a^2+\zo^2)}
\right]^2+\left[\frac{d_2-d_2^*}{4ia}\right]^2\nonumber\\
&=\frac{(d_1+d_1^*)^2+(d_2+d_2^*)^2}{16(a^2+\zo^2)}
-\frac{(d_1-d_1^*)^2+(d_2-d_2^*)^2}{16a^2}
\end{align*}
is positive. Therefore,
\begin{align*}
q_l&=q_r+\frac{2\begin{pmatrix}c_1&-c_2^*\end{pmatrix}}{\det\bP}\times\\
&\times\begin{pmatrix}-\tfrac{d_1-d_1^*}{4a}+\tfrac{a(d_1+d_1^*)}{4(a^2+\zo^2)}
+\tfrac{\zo(d_2+d_2^*)}{4(a^2+\zo^2)}&-\tfrac{d_2-d_2^*}{4a}
+\tfrac{a(d_2+d_2^*)}{4(a^2+\zo^2)}-\tfrac{\zo(d_1+d_1^*)}{4(a^2+\zo^2)}\\
-\tfrac{d_2-d_2^*}{4a}-\tfrac{a(d_2+d_2^*)}{4(a^2+\zo^2)}
+\tfrac{\zo(d_1+d_1^*)}{4(a^2+\zo^2)}&\tfrac{d_1-d_1^*}{4a}
+\tfrac{a(d_1+d_1^*)}{4(a^2+\zo^2)}+\tfrac{\zo(d_2+d_2^*)}{4(a^2+\zo^2)}
\end{pmatrix}\begin{pmatrix}-b_2^*\\b_1^*\end{pmatrix}.
\end{align*}
\end{example}

\section*{Acknowledgments}
The authors have been partially supported by the Regione Autonoma della Sardegna
research project {\it Algorithms and Models for Imaging Science} [AMIS]
(RASSR57257, intervento finanziato con risorse FSC 2014-2020 -- Patto per lo
Sviluppo della Regione Sardegna), and by INdAM-GNFM.

\appendix

\section{Invertibility of the Sylvester solutions}\label{sec:A}

Given a matrix triplet $(\bA,\bB,\bC)$, where $\bA$ is a $2p\times2p$ matrix
whose eigenvalues have positive real parts, $\bB$ is a $2p\times r$ matrix, and
$\bC$ is an $r\times2p$ matrix, we define the {\it controllability subspace} and
the {\it observability subspace} of $\C^{2p}$ as follows:
\begin{subequations}\label{C.1}
\begin{align}
\text{Im}(\bA,\bB)&=\bigvee_{j=0}^\infty\,\text{Im}(\bA^j\bB),
\label{C.1a}\\
\text{Ker}(\bC,\bA)&=\bigcap_{j=0}^\infty\,\text{Ker}(\bC\bA^j),\label{C.1b}
\end{align}
\end{subequations}
where $\text{Im}\,T$ and $\text{Ker}\,T$ stand for the range and the null space
of a matrix $T$, respectively. The $V$-symbol in \eqref{C.1a} denotes the
set of finite linear combinations of vectors in the union of
$\text{Im}\,(\bA^j\bB)$ ($j=0,1,2,\ldots$) and the intersection in \eqref{C.1b}
is finite. We observe that $\text{Im}(\bA,\bB)$ is the smallest $\bA$-invariant
subspace containing $\text{Im}\,\bB$ and $\text{Ker}(\bC,\bA)$ is the largest
$\bA$-invariant subspace contained in $\text{Ker}\,\bC$. We call the matrix pair
$(\bA,\bB)$ {\it controllable} if $\text{Im}(\bA,\bB)=\C^{2p}$. We call the
matrix pair $(\bC,\bA)$ {\it observable} if $\text{Ker}(\bC,\bA)$ is the zero
subspace. The matrix triplet $(\bA,\bB,\bC)$ is called {\it minimal} if
$(\bA,\bB)$ is controllable and $(\bC,\bA)$ is observable [or: if $\bA$ has
minimal matrix order among the triplets $(\bA,\bB,\bC)$ leading to the same
$\zO(z)=\bC e^{-z\bA}\bB$].
A comprehensive account of controllability and observability can be found in any
textbook on linear control theory \cite{BGK,CZ,HRS}.

 For the above matrix triplets we obviously have in mind $(\bA,\bB_r,\bC_r)$ and
$(\bA,\bB_l,\bC_l)$. In most of this subsection we drop the subscripts $r$ and
$l$ and consider the triplets $(\bA,\bB\bC,I_{2p})$ and $(\bA,I_{2p},\bB\bC)$
with $r=2p$ as well.

The next result relies on arguments provided by Hearon \cite{H} for triplets
$(\bA,\bB,\bC)$ of complex matrices. Here Hearon's arguments are adapted to
matrix triplets $(\bA,\bB,\bC)$, where $\bA\in\bzS^{p\times p}$,
$\bB\in\bzS^{p\times1}$, and $\bC\in\bzS^{1\times p}$.

\begin{theorem}\label{th:A.1}
Let $(\bA,\bB,\bC)$ be a matrix triplet, where $\bA\in\bzS^{p\times p}$ only has
eigenvalues with positive real part, $\bB\in\bzS^{p\times1}$, and
$\bC\in\bzS^{1\times p}$. Then the following statements are equivalent:
\begin{itemize}
\item[$(a)$] The unique solution $\bP$ of the Sylvester equation
\begin{equation}\label{A.2}
\bA\bP+\bP\bA=\bB\bC
\end{equation}
is invertible.
\item[$(b)$] The pair $(\bA,\bB\bC)$ is controllable.
\item[$(c)$] The pair $(\bB\bC,\bA)$ is observable.
\end{itemize}
\end{theorem}

\begin{proof}
Let us first prove that $\text{Im}\,(\bB\bC)$ is contained in $\text{Im}\,\bP$
iff $\text{Im}\,\bP$ is $\bA$-invariant. Indeed, if $\text{Im}\,\bP$ is
$\bA$-invariant, then for each $\bh\in\C^{2p}$ there exists $\bk\in\C^{2p}$ such
that $\bA\bP\bh=\bP\bk$; then, using \eqref{A.2}, we get
$\bB\bC\bh=\bP(\bk+\bA\bh)$, thus proving that $\text{Im}\,(\bB\bC)$ is
contained in $\text{Im}\,\bP$. Conversely, if $\text{Im}\,(\bB\bC)$ is contained
in $\text{Im}\,\bP$, then for each $\bh\in\C^{2p}$ there exists $\bk\in\C^{2p}$
such that $\bB\bC\bh=\bP\bk$; then, using \eqref{A.2}, we get
$\bA\bP\bh=\bP(\bk-\bA\bh)$, thus proving that $\text{Im}\,\bP$ is
$\bA$-invariant.

Next, we prove that $\text{Ker}\,(\bB\bC)$ contains $\text{Ker}\,\bP$ iff
$\text{Ker}\,\bP$ is $\bA$-invariant. Indeed, if $\text{Ker}\,\bP$ is contained
in $\text{Ker}\,(\bB\bC)$, then for each $\bh\in\C^{2p}$ such that $\bP\bh=0$
we have $\bB\bC\bh=0$, which implies that $\bP\bA\bh=\bB\bC\bh-\bA\bP\bh=0$.
Conversely, if $\text{Ker}\,\bP$ is $\bA$-invariant, then for each
$\bh\in\C^{2p}$ such that $\bP\bh=0$, we have $\bP\bA\bh=0$ and hence
$\bB\bC\bh=\bA\bP\bh+\bP\bA\bh=0$.

(b)$\Rightarrow$(a). Let $\bq$ be a $p\times1$ matrix with entries in $\bzS$
such that $\bP\bq=0$. Then $\bP\bA\bq=\bB\bC\bq$. Then there are two options:
\begin{itemize}
\item[(i)] $\bC\bq=0$ whenever $\bP\bq=0$, or
\item[(ii)] $\bC\bq\neq0$ for some $\bq$ satisfying $\bP\bq=0$.
\end{itemize}
In the first case, we see that $\bP\bA\bq=\bB\bC\bq-\bA\bP\bq=0$ and hence the
kernel of $\bP$ is $\bA$-invariant. If we then also assume that $(\bB\bC,\bA)$
is observable, then the $\text{Im}\,\bP$ contains the smallest $\bA$-invariant
subspace containing $\text{Im}\,(\bB\bC)$ and hence the controllability of
$(\bA,\bB\bC)$ implies that $\bP$ is invertible. In the second case we see that
the $\bzS$-vector $\bB$ belongs to the range of $\bP$, implying that the range
of $\bB\bC$ is contained in the range of $\bP$ so that the range of $\bP$ is
$\bA$-invariant. If we then also assume that $(\bA,\bB\bC)$ is controllable and
hence the smallest $\bA$-invariant subspace containing the range of $\bB\bC$ is
all of $\bzS^{p\times1}$, then $\bP$ is invertible. In either case we conclude
that $\bP$ is invertible.

(c)$\Rightarrow$(a). Using the arguments of the preceding paragraph, we see that
the controllability of $(\bA^\dagger,\bC^\dagger\bB^\dagger)$ implies the
invertibility of $\bP^\dagger$.

(a)$\Rightarrow$[(b)+(c)]~Let us first assume $\bP$ to be invertible. To
prove the controllability of the pair $(\bA,\bB\bC)$, we take a vector
$\bh\in\C^{2p}$ orthogonal to the controllability subspace
$\text{Im}\,(\bA,\bB\bC)$. Then
$$(\bA^j\bB\bC\bP^{-1}\bk,\bh)=0,\qquad\bk\in\C^{2p},\quad j=0,1,2,\ldots.$$
Therefore, using the identity [cf. \eqref{A.2}]
\begin{equation}\label{A.3}
\bA=\bB\bC\bP^{-1}-\bP\bA\bP^{-1},
\end{equation}
for arbitrary $\bk\in\C^{2p}$ and $j=0,1,2,\ldots$ we have
\begin{align*}
(\bA^{j+1}\bk,\bh)&=(\bA^j\bB\bC\bP^{-1}\bk,\bh)-(\bA^j\bP\bA\bP^{-1}\bk,\bh)\\
&=-(\bA^j\bP\bA\bP^{-1}\bk,\bh).
\end{align*}
By the arbitrariness of $\bk$ we get
\begin{equation}\label{A.4}
{\bA^\dagger}^{j+1}\bh
=-{\bP^\dagger}^{-1}\bA^\dagger\bP^\dagger{\bA^\dagger}^j\bh.
\end{equation}
Repeated application of \eqref{A.4} yields
\begin{align*}
{\bA^\dagger}^j\bh&=-{\bP^\dagger}^{-1}\bA^\dagger\bP^\dagger{\bA^\dagger}^{j-1}
\bh=({\bP^\dagger}^{-1}\bA^\dagger\bP^\dagger)^2{\bA^\dagger}^{j-2}\bh\\
&=\ldots=(-{\bP^\dagger}^{-1}\bA^\dagger\bP^\dagger)^j\bh
={\bP^\dagger}^{-1}(-\bA^\dagger)^j\bP^\dagger\bh,
\end{align*}
which implies that
$$q(\bA^\dagger)\bh={\bP^\dagger}^{-1}q(-\bA^\dagger)\bP^\dagger\bh$$
for any polynomial $q(z)$. If we take $q(z)=\det(zI_{2p}-\bA^\dagger)$ [the
characteristic polynomial of $\bA^\dagger$], we obtain $q(\bA^\dagger)=0$ by the
Cayley-Hamilton theorem \cite{G}. Using that $\bA^\dagger$ and $-\bA^\dagger$ do
not have common eigenvalues [and hence $q(z)$ and $q(-z)$ do not have common
zeros], we obtain the invertibility of $q(-\bA^\dagger)$. Consequently, $\bh=0$.
As a result, $\text{Im}\,(\bA,\bB\bC)=\C^{2p}$, yielding the controllability of
the pair $(\bA,\bB\bC)$. Finally, using the invertibility of $\bP^\dagger$,
we prove the controllability of the pair $(\bA^\dagger,\bC^\dagger\bB^\dagger)$
and hence the observability of the pair $(\bB\bC,\bA)$.
\end{proof}

\begin{corollary}\label{th:A.2}
The matrices $\bP_r$ defined by {\rm\eqref{5.7}} and $\bP_l$ defined by
{\rm\eqref{5.10}} are invertible.
\end{corollary}

\begin{theorem}\label{th:A.3}
 For each $x\in\R$ except at finitely many values, the matrices
$e^{2x\bA}+\bP_r$ and $e^{-2x\bA}+\bP_l$ are invertible.
\end{theorem}

Example~\ref{th:6.1} contains a triplet where $\det(e^{2x\bA}+\bP_r)=0$ for
some $x\in\R$.

\begin{proof}
In Theorem \ref{th:4.1} above we have proved the nonnegativity of the
determinants of $\bP_{r,l}$ and $e^{\pm2x\bA}e^{-t\bH}+\bP_{r,l}$ for each
$x\in\R$. Since for each $t\in\R$ the function
$\det(e^{\pm2x\bA}e^{-t\bH}+\bP_{r,l})$ is entire analytic in $x$, is
nonnegative on the real $x$-line, tends to $+\infty$ as $x\to\pm\infty$ along
the real line, and tends to $\det\bP_{r,l}>0$ as $x\to\mp\infty$, there are at
most finitely values of $x\in\R$ for which the matrix
$e^{\pm2x\bA}e^{-t\bH}+\bP_{r,l}$ is singular.
\end{proof}

\end{document}